\documentclass[10pt, a4paper, twocolumn]{article}
\usepackage[cm]{fullpage}
\usepackage[english]{babel}
\usepackage{xcolor}
\usepackage{blindtext}
\usepackage{xspace}
\usepackage{amsthm}
\usepackage{amsmath}
\usepackage{amssymb} 
\usepackage{hyperref}
\usepackage{graphicx}
\usepackage{xspace, soul}
\usepackage{mathtools}
\usepackage{pdfcomment}
\usepackage{url}
\usepackage{enumitem}
\usepackage{makecell}
\usepackage{caption}
\usepackage{subcaption}
\usepackage{algorithm}
\usepackage{algpseudocode}
\usepackage{float}
\usepackage{balance}
\usepackage{thmtools, thm-restate}
\usepackage{breqn}

\newcommand{\am}[1]{{\color{magenta} \footnotesize[Ashwin: #1] }}

\newcommand{\yt}[1]{{\color{blue} \footnotesize[Yuchao: #1] }}

\newcommand{\eat}[1]{}

\newtheorem{definition}{Definition}[section]

\newcommand\numberthis{\addtocounter{equation}{1}\tag{\theequation}}

\newcommand{\norm}[1]{\left\lVert#1\right\rVert}
\newcommand{\Lagr}{\mathcal{L}}

\DeclareMathOperator*{\argmin}{arg\,min}

\newcommand{\paratitle}[1]{\textbf{#1}}

\newcommand{\D}{D}
\newcommand{\A}{A}

\newcommand{\dom}{\Sigma}
\newcommand{\tup}{t}
\newcommand{\nY}{\tilde{y}}
\newcommand{\ystar}{y^{*}}
\newcommand{\qstar}{q^{*}}
\newcommand{\pp}{p} 
\newcommand{\qq}{q} 
\newcommand{\f}{h}

\newcommand{\Wstate}{W_S}
\newcommand{\Wgroup}{W_G}

\newcommand{\round}[1]{\ensuremath{\lfloor#1\rceil}}

\newcommand{\squishlist}{
	\begin{list}{$\bullet$}
		{
			\setlength{\itemsep}{0pt}
			\setlength{\parsep}{3pt}
			\setlength{\topsep}{3pt}
			\setlength{\partopsep}{0pt}
			\setlength{\leftmargin}{1.5em}
			\setlength{\labelwidth}{1em}
			\setlength{\labelsep}{0.5em} } }

\newcommand{\squishend}{
\end{list}  }

\title{Prior-Aware Distribution Estimation for Differential Privacy}
\author{Yuchao Tao, Johes Bater, Ashwin Machanavajjhala}
\date{
Duke University
}

\begin{document}

\maketitle

\begin{abstract}
    Joint distribution estimation of a dataset under differential privacy is a fundamental problem for many privacy-focused applications, such as query answering, machine learning tasks and synthetic data generation. In this work, we examine the joint distribution estimation problem given two data points: 1) differentially private answers of a workload computed over private data and 2) a prior empirical distribution from a public dataset. Our goal is to find a new distribution such that estimating the workload using this distribution is as accurate as the differentially private answer, and the relative entropy, or KL divergence, of this distribution is minimized with respect to the prior distribution. We propose an approach based on iterative optimization  for solving this problem. An application of our solution won second place in the NIST 2020 Differential Privacy Temporal Map Challenge, Sprint 2.
\end{abstract}

\section{Introduction}
Personal information collected as flat tabular data is used in many real world applications, such as medical records, community survey records and web visit records \cite{dankar2013practicing, ruggles2018ipums, fan2014monitoring}. Releasing these sensitive records to the public enables a variety of scientific research, but introduces risk due to privacy leakage. Differential privacy (DP) is a state-of-the-art technique used to prevent privacy leakage \cite{dwork2006calibrating}. Releasing a synthetic dataset under differential privacy provides public access to data along with strong privacy guarantees.

A common way to synthesize tabular data under DP is to privately learn the joint distribution by injecting noise according to a DP mechanism, such as the Laplace mechanism. When using the Laplace mechanism, we inject noise to each domain value, which results in extremely high error when the total domain size is large, which is common for datasets with multiple attributes. To mitigate this problem, many probabilistic-graphical-model based solutions are proposed, such as decomposing the joint distribution into a product of multiple small marginal distributions through either a Bayesian \cite{zhang2017privbayes} or Markov network \cite{chen2015differentially, bernstein2017differentially, mckenna2019graphical, zhang2020privately}. Another line of approach is to privately learn the joint distribution through estimation. One example, MWEM \cite{hardt2010simple}, updates the joint distribution according to the observed noisy marginal distribution using the multiplicative rule. Another example, PGM \cite{mckenna2019graphical}, learns the joint distribution on a graphical model from a marginal manifold while maximizing the entropy of non-observed distributions. All of these techniques are based on first estimating answers to linear queries under DP and then constructing a synthetic dataset from that.

However, all of these approaches assume no knowledge about the data before learning the joint distribution. Given a sensitive dataset for a specific domain, similar datasets can often be found that have been released publicly in the past.  We can treat the empirical joint distribution from the public dataset as a prior for our private joint distribution, which contains more information that can not be learned from the DP answers of a workload.
\eat{
When we have the noisy answer of a workload due to the DP mechanism, 
 we not only use it to estimate the private joint distribution, but also take the public prior into consideration.
 }

In this work, we consider the joint distribution estimation problem given a prior distribution and the answers to a linear query workload under DP. Our goal is to find a joint distribution such that 
the answers to the given workload derived from this distribution are sufficiently accurate, and the relative entropy, or KL divergence, of this distribution with respect to the prior is minimized. We believe this approach can improve the quality of the synthetic dataset sampled from the inferred joint distribution to evaluate more queries accurately even outside the given workload, as long as the prior captures the information precisely that we cannot learn from the DP answers of the given workload.

A similar approach is found in PMWPub \cite{liu2021leveraging}. It combines MWEM with public data by initializing the joint distribution with the public distribution and updating it using the multiplicative rule according to the noisy answer of the query selected through exponential mechanism from a fixed query class. The goal for PMWPub is to have a synthetic dataset such that any query from the fixed query class can be evaluated accurately. Unlike PMWPub, our goal is to estimate the joint distribution such that not only the given workload can be evaluated accurately, but also the relative entropy of the new distribution with respect to the prior is low.

Our contributions are as follows:
\squishlist{
\item We initiate the study of joint distribution estimation given a public prior and DP answers to a workload of linear queries.
\item The estimated distribution has two characteristics: 1) it has low $L_1$ error with respect to the workload. 2) it has low relative entropy with respect to the prior.
\item The estimated distribution is tractable since its support is covered by the support of the public prior, which is derived from a public dataset and is tractable by nature.
\item An application of this approach was the second place winning solution in the NIST 2020 Differential Privacy Temporal Map Challenge, Sprint 2 \cite{deid2}.
}

\squishend

\section{Preliminaries}

\paratitle{Data.} We consider a flat tabular data $\D$ with $m$ attributes and $n$ tuples. The domain for each attribute $\A_i$ is $\dom_i$, and the domain for a tuple $\tup$ in $\D$ is $\dom = \dom_1 \times \dom_2 \times \ldots \dom_m$. We assume each attribute domain is a discrete domain \footnote{For numerical domains, it is practical to discretize the domain by some natural partitions.}. Denote $N$ as the size of domain $\dom$. We consider a natural ordering of all the domain values in $\dom$ such that $\dom = \{v_1, v_2, ..., v_N\}$.  
We use $\f = [\f_1, \f_2, \ldots \f_N]$ to denote the frequency for each domain value $v_i$ in $\D$, which is a histogram on the domain $\dom$ for $\D$.
We further denote $x = [x_1, x_2, \ldots x_N]$ where $x_i = \f_i / n$ as the empirical density distribution of the data $\D$.  
 We also use the notation $x(v_i)$ to refer $x_i$.

\paratitle{Workload.} A linear query $u$ is a vector of size $N$, such that the answer is given as $u^Tx$. If $u$ only contains $\{0, 1\}$, it asks the empirical probability that a tuple in $\D$ satisfies any of the predicates defined by $u$. A workload $W$ of $M$ linear queries is a matrix of size $M \times N$, such that the answer is given as $y = Wx$. We assume $W$ only contains $\{0, 1\}$, which means all queries are empirical probability queries for some predicates. A $k$-way marginal query is a workload $W$ such that it asks the marginal probability empirically for some $k$ attributes. Denote $u[i]$ as the $i$-th entry in the vector $u$ and $W[i, j]$ as the entry at $i$-th row and $j$-th column. Denote $W[i, *]$ as the $i$-th row in $W$ and $W[*, j]$ as the $j$-th column in $W$. We use $\norm{\cdot}_p$ as the $L_p$ norm. Notice that for a matrix $W$, $\norm{W}_1$ is the max entry and $\norm{W}_{1,1}$ is the sum of absolute values of all entries.

\paratitle{Differential Privacy.} 
We assume the workload is answered under differential privacy.
\begin{definition}[Differential Privacy]
A randomized mechanism $M: \dom^* \rightarrow \mathcal{R}$ satisfies $\epsilon$-differential privacy (DP) if for any neighboring datasets $D$ and $D'$ which differ by a single tuple, and every event $O \subseteq \mathcal{R}$, we have:
\[\ln{\frac{P(M(D) \in O)}{P(M(D') \in O)}} \leq \epsilon \]
\end{definition}
A typical differential private mechanism is Laplace mechanism, which adds Laplace noise to the workload answers with the scale proportional to the global sensitivity that is the maximum change of the workload answers by changing one row in the data  \cite{dwork2006calibrating}. We denote $\nY$ as the mechanism output about answering $y = Wx$.

\eat{
\begin{definition}[Global Sensitivity]
Given a query $Q: \dom^* \rightarrow R^k$, its global sensitivity $\Delta$ is defined as 
\[\Delta(Q) = \max_{D \sim D'} |Q(D) - Q(D')| \]
where $D \sim D'$ denotes all neighbouring datasets.
\end{definition}

\begin{definition}[Laplace Mechanism]
Given a query $Q: \dom^* \rightarrow R^k$ with global sensitivity $\Delta(Q)$, Laplace mechanism satisfies $\epsilon$-DP by adding Laplace noise to the answer as:
\[M(D) = Q(D) + Lap(\Delta(Q) / \epsilon)\]
where $Lap(\Delta(Q) / \epsilon)$ samples $k$ i.i.d. random values from the Laplace distribution with scale equal to $\Delta(Q) / \epsilon$.
\end{definition}

Since the global sensitivity for a workload $W$ is $\norm{W}_1/n$, where $\norm{\cdot}_1$ denotes the maximal column sum, we can use Laplace mechanism to learn the marginal distribution by adding Laplace noise scaled to $\norm{W}_1/(n\epsilon)$ and satisfy $\epsilon$-DP.
}

\paratitle{Prior Distribution}
We assume a public dataset with the same domain of the private dataset is available. Usually, the mutual information between the public dataset and the private dataset is not trivial. We consider the empirical density distribution from the public dataset as a prior $\pp$, which is also a vector of size $N$. 

\begin{definition}[KL Divergence \cite{kullback1997information}]
Given two distributions $q$ and $p$ on the same domain $\dom$, the KL divergence is defined as
\begin{align}
D_{KL}(q || p) = \sum_{v_i \in \dom} q(v_i) \log \frac{q(v_i)}{p(v_i)}
\end{align}
. This is also called the relative entropy of $q$ with respect to $p$.
\end{definition}

\paratitle{Problem Statement.}
Given a private dataset $\D$ with $n$ tuples and a workload $W$ of queries about $x$, the empirical density distribution of $\D$,  we assume $W$ is answered by some mechanism that satisfies $\epsilon$-DP, such as Laplace mechanism \cite{dwork2006calibrating} and HDMM \cite{mckenna2018optimizing}, and denote the noisy answer as $\nY$. 
We also assume a prior distribution $\pp$ is given, which is the empirical density distribution from a public dataset. The goal is to find a new distribution $\qq$ such that the followings are satisfied.
\squishlist
 \item 1) Estimating $W$ from $\qq$ is at least as accurate as $\nY$; i.e., 
 we want to minimize the $L1$ distance between $W\qq$ and $\nY$. 
\item   2) Given the constraint that the distribution $\qq$ must minimize the $L_1$ distance between $W\qq$ and $\nY$, we want $\qq$ to be as similar as possible to the prior $\pp$; i.e., the relative entropy of $\qq$ with respect to the prior $\pp$, or KL divergence $D_{KL}(\qq || \pp)$, is low, according to the minimum relative entropy principle (MRE) \cite{woodbury1993minimum, olivares2007quantum, hou2005minimum, muhlenbein2005estimation}. When $\pp$ is a uniform prior, minimizing the KL divergence $D_{KL}(\qq || \pp)$ is equivalent to the principle of maximum entropy \cite{jaynes1957information}. 
\squishend


\section{Distribution Estimation}

In this work, we consider two sources of information that can help us to estimate the distribution of the private data: the prior distribution $\pp$ and the noisy answer $\nY$ of the workload $W$. 
Since $\nY$ is derived directly from the private data and the prior $\pp$ is purely an initial guess, we treat $\nY$ as the first-class information and $\pp$ as the second-class information. That is, when we estimate the private distribution $\qq$, we want $W \qq$ to be similar to $\nY$ as much as possible. If there is nothing more we can learn from the noisy answer $\nY$, we then learn from the prior $\pp$.

\subsection{Two-stage Optimization}

Given two optimization goals for the distribution estimation problem, we divide the optimization into two stages. In the first stage, we aim to find a set of valid distribution $\qq$ such that the $L_1$ distance between $W\qq$ and $\nY$ is minimized.
Denote $\ystar = \argmin_{y} \{\norm{y - \nY}_1 \mid \exists z, \norm{z}_1 = 1, z \geq 0, Wz = y\} $. We argue that $W\qq = \ystar$ is the sufficient and necessary condition that $\norm{W\qq - \nY}_1$ is minimized.

\eat{
e minimize th e loss of $\qq$ as L1 error regarding to $\nY$:
\begin{align*}
\min ~& \norm{Wz - \nY}_1 \\
s.t. ~& \norm{z}_1 = 1 \\
     ~& z \geq 0
\end{align*}
and denote the minimum as $\ystar$. \am{We are finding a $z$. So should this be $z^\star$?}\yt{Actually the only interesting part is the optimal objective, not the optimal value of the variable. Maybe it is confusing in this form?}  \am{If the noisy answers are inconsistent, does this step result in consistency? } \yt{I think it is true.}
}

In the second stage, we minimize the relative entropy of $\qq$ with respect to the prior $\pp$ with the constraint that $\qq$ must satisfy $W\qq = \ystar$, which follows the MRE principle:
\begin{align*}
\min ~&  \sum_{v_i \in \dom} \qq(v_i) \log \frac{\qq(v_i)}{\pp(v_i)} \\
s.t. ~& W\qq = \ystar  \\
~& \sum_{v_i \in \dom} \qq(v_i) = 1 \\
~& \forall v_i \in \dom, \qq(v_i) \geq 0 \numberthis
\end{align*}
and denote $q$ as $\qstar$ when the minimum is achieved. This is the final estimation of the private distribution.

\subsection{Optimization Implementation}

The first optimization problem is a linear programming problem, which can be solved efficiently if the size of $z$ is small, which is the domain size of the data $\D$. When the domain size is large, we can narrow down the domain of $\qq$ to the support from the prior $\pp$, which could be much smaller than the original domain size;
i.e., $\ystar = \argmin_{y} \{\norm{y - \nY}_1 \mid \exists z, \norm{z}_1 = 1, z \geq 0, \left(\forall v_i \in \dom , \pp(v_i) = 0 \implies z(v_i) = 0 \right), Wz = y\}$. 

For the second optimization problem, one can apply the method of Lagrange multipliers to solve:
\begin{dmath*}
    \Lagr(q, \lambda, \mu) = 
     \sum_{v_i \in \dom} \qq(v_i) \log \frac{\qq(v_i)}{\pp(v_i)}
     - \lambda^T(W\qq - \ystar) - \mu(\norm{\qq}_1 - 1)  
\end{dmath*}
where the vector $\lambda$ and $\mu$ are Lagrange multipliers.
The optimal solution for the original problem is thus given by the values when $\nabla{\Lagr}=0$. Solving $\nabla{\Lagr}=0$ gives \cite{woodbury1993minimum}
\begin{align}
    \qq(v_i) = \pp(v_i) \exp\left[ -1 -\mu - \lambda^T W[*, i] \right]
\end{align}
\eat{
or
\begin{align}
\mu + \lambda^T W[*, i] = - \ln\left(\frac{\qq(v_i)}{\pp(v_i)} \right) - 1
\end{align}
}
Since KL divergence is convex and the workload $W$ is linear, the solution is unique (\cite{kapur1989maximum}, chapter 8). However, it is not clear to have a closed form about $\lambda$ and $\mu$.

In the next section, we propose an iterative optimization approach to achieve both optimization goals, which improves the computational efficiency of solving the linear programming problem and bypasses solving the Lagrange multipliers from the equations.

\subsection{Iterative Optimization}

We consider merging two optimizations into one:
\begin{align*}
    \min ~& \norm{W\qq - \nY}_1\\
    s.t. ~& \qq \in \underset{z \in \mathcal{S}}{\argmin}\left( \sum_{v_i \in \dom} z(v_i) \log \frac{z(v_i)}{\pp(v_i)}\right) \\
         ~& \mathcal{S} = \{z \mid \norm{Wz - \nY}_1 = \norm{W\qq - \nY}_1, \\
         ~& \qquad \quad \;\;\, \norm{z}_1 = 1, z \geq 0\} \numberthis
\end{align*}
, which finds the optimal distribution $\qstar$ such that $\norm{W\qstar - \nY}_1$ is minimized first and then the relative entropy of $\qstar$ with respect to the prior $\pp$ is also minimized.

Based on this observation, we consider an iterative approach to find the optimal. At the beginning, we set $\qq$ as $\pp$. For each iteration, we select some queries from the workload with $L_1$-norm equal to 1 as a partial workload $W_s$, with corresponding partial noisy answers $\nY_s$. A typical partial workload $W_s$ with $L_1$-norm equal to 1 is a marginal query. We then find a new positive answer $y_s$ to the partial workload, such that $\norm{y_s - \nY_s}_1$ is minimized. The sum of $y_s$ should be equal to $1$ if $W_s$ covers the entire domain,  otherwise it should be within $1$. We then update $\qq$ according to $W_s$ and $y_s$ such that the relative entropy of $\qq$ with respect to the prior $\pp$ is minimized while we ensure $W_s \qq = y_s$. The details of distribution update is defined below. 
After $\qq$ is updated, we then move on to the next iteration by treating $\qq$ as the new prior, until the limit of iterations is met. 

\begin{definition}[Distribution Update Function]\label{def:distr_update} Given a prior $\pp$, a partial workload $W_s$ with $L_1$-norm equal to 1 and the corresponding optimal answer $y_s$, the function $f$ defined as follows is a distribution such that for each domain value $v_j \in \dom$ :
\begin{align*}
    &f(p; W_s, y_s)[j] = \\
    &\begin{cases}
    \frac{p[j] y_s[i]}{W_s[i, *]^Tp}, 
    & \text{if}~ \exists i, W_s[i, j] > 0 ~\&~ W_s[i, *]^Tp > 0 \\
    \frac{y_s[i]}{\norm{W_s[i, *]}_1}, 
    & \text{if}~ \exists i, W_s[i, j] > 0 ~\&~ W_s[i, *]^Tp = 0 \\
    \frac{p[j](1 - \norm{y_s}_1)}{1 -\norm{W_sp}_1}, 
    & \text{if}~ \forall i, W_s[i, j] = 0 ~\&~ 1 -\norm{W_sp}_1 > 0 \\
    \frac{(1 - \norm{y_s}_1)}{N -\norm{W_s}_{1,1}},  
    & \text{if}~ \forall i, W_s[i, j] = 0 ~\&~ 1 -\norm{W_sp}_1 = 0 \\
    \end{cases}
\end{align*}
\end{definition}

\begin{restatable}[Least Relative Entropy Update]{theorem}{lreu}
\label{thm:lreu}
Given a prior $\pp$, a partial workload $W_s$ with $L_1$-norm equal to 1 and the corresponding optimal answer $y_s$, the distribution update function $f(\pp; W_s, y_s)$ finds the distribution that is of minimal relative entropy with respect to the prior $\pp$.
\end{restatable}

\begin{algorithm}[H] 
\caption{Iterative Distribution Estimation} 
\begin{algorithmic}[1]
\Require A $M \times N$ workload $W$, the corresponding DP answer $\nY$ based on a private dataset, the public prior $\pp$, and an iteration parameter $T$.
\Ensure The estimated distribution $\qq$ of the underlying private data.

\Function{IDE}{$W$, $\nY$, $\pp$, $T$}
\State $\qq^{(0)} \gets \pp$ 
\For{$t \gets 1 \ldots T$}
\State Randomly select a partial workload $W_s$ from $W$ such that $\norm{W_s}_1 = 1$.
\State Let $\nY_s$ be the corresponding answers from $\nY$ regarding to $W_s$.
\If {$\norm{W_s}_{1,1} = N$}
    \State $y_s \gets \argmin_{\norm{y_s}_1 = 1, y_s \geq 0}\{\norm{y_s - \nY_s}_1\}$
\Else
    \State $y_s \gets \argmin_{\norm{y_s}_1 \leq 1, y_s \geq 0}\{\norm{y_s - \nY_s}_1\}$ 
\EndIf
\State $\qq^{(t)} \gets f(\qq^{(t-1)}; W_s, y_s)$ 
\Comment{See Def \ref{def:distr_update}}
\EndFor
\State \Return $\qq^{(T)}$.
\EndFunction
\end{algorithmic}
\label{alg:iterative}
\end{algorithm}

Algorithm \ref{alg:iterative} illustrates the flow of the iterative approach to find $\qstar$. After each iteration, we reduce the loss about $\norm{W\qq - \nY}_1$, while the rule of distribution update also keeps the updated distribution to be close to the previous distribution in the sense of relative entropy. After sufficient iterations, the loss $\norm{W\qq - \nY}_1$ of the final distribution should be close to the minimal, and the relative entropy of $\qq$ with respect to the prior $\pp$ as $D_{KL}(\qq || \pp)$ is low. It is not yet clear how convergent this iterative algorithm is.

To boost the computational efficiency, we can narrow down the domain of $\qq$ as the support of the public prior $\pp$ for the algorithm. After the algorithm is finished, we can implicitly extend the distribution of $\qq$ to the full domain such that $\qq(v_i) = 0$ if $\pp(v_i) = 0$. 


\section{Application}

In ``2020 Differential Privacy Temporal Map Challenge" \cite{deid2}, an application of our approach was the second place winning solution for the sprint 2.  This challenge is basically to generate a synthetic dataset such that its k-way marginals has low error.

\paratitle{Data.} The data considered in the contest is a subset of IPUMS American Community Survey data \cite{ruggles2018ipums} about demographic and financial features for some states across some years. The data about states Ohio and Illinois from 2012-2018 is given as a public dataset. Each row of data includes the geographical and temporal features: PUMA and YEAR, and other 33 additional survey features, such as the biological features, work-related features and financial features. The privacy is at the individual level, and we define individual stability $\Delta$ as the max number of records that is associated with a single individual.

\paratitle{Measurement.} 
The error of the synthetic dataset is measured as follows. Both the synthetic dataset and the ground-truth data are first grouped by the PUMA and YEAR. For each group, we randomly choose two attributes other than PUMA and YEAR, generate the marginal distributions empirically from the synthetic and the ground-truth dataset, and compute the $L_1$ distance between two marginal distributions, which ranges from 0 to 2. This is repeated for a fixed number of times for each group, and the error for each group is the average $L_1$ distance across repetitions. If the total count of a group differs from the true count by 250, the error of that group will be set to 2 as the bias penalty. The final error is the average of errors from all groups.

\subsection{Algorithm}

The main idea for our approach is based on ``select-measure-reconstruct" \cite{mckenna2018optimizing}. We first select some marginal queries, measure them by Laplace mechanism, and then reconstruct the joint distribution by updating the public prior using Algorithm \ref{alg:iterative} ``Iterative Distribution Estimation". We define it as a subroutine \textproc{Prior\_Update} in Algorithm \ref{alg:sdg} (see Appendix), as it takes a prior as input and outputs an updated prior. An updated prior can be further updated by calling \textproc{Prior\_Update} again. We apply this idea in Algorithm \ref{alg:sdg}. Here we consider grouping the private data $D$ by PUMA and YEAR. Each group is expressed as a triple $(puma, year, \f)$, where $\f$ is the frequency distribution.  At the first round, we re-group the data by STATE, which is inferred from the PUMA attribute, and then apply the subroutine to each STATE partition to update the prior. This is because we believe the joint distributions of all PUMA-YEAR groups in the same STATE are similar, and we have sufficient data to update the prior accurately. The workload  chosen for the state-level data is denoted as $\Wstate$. Now we get a new prior for each STATE. Within each STATE, we run this subroutine for each PUMA-YEAR group with the new prior. The workload chosen for the group-level data is denoted as $\Wgroup$. 

\begin{restatable}{theorem}{sdg}
\label{thm:sdg} Algorithm \ref{alg:sdg} satisfies $\epsilon$-DP.
\end{restatable}

\subsection{Parameter Tuning}
We also do a heavy parameter tuning for the contest, especially for the choice of workloads.

\paratitle{Workload $\Wstate$ and $\Wgroup$.} We only consider marginal queries for our workloads. 
The goal is to find marginal queries such that after we estimate the distribution based on the DP answers of these marginal queries, the total error of all 2-way marginal queries is low. This problem can be viewed as selecting a strategy workload for the target workload with all 2-way marginals in the framework of HDMM \cite{mckenna2018optimizing}. Unlike HDMM, of which the strategy selection is data independent, we consider a data dependent strategy to select the workload based on the public dataset.
Given the public dataset, we compute the pair-wise mutual information \cite{steuer2002mutual} for all attributes, and covert it into a graph by adding edges with high mutual information. We then select some cliques from the graph as the workload $\Wstate$ for the state-level data, based on the idea of learning distribution by graphical model \cite{zhang2017privbayes, chen2015differentially, bernstein2017differentially, mckenna2019graphical, zhang2020privately}. For data of each PUMA-YEAR group, since the data size is much smaller, we cannot learn many high-way marginals accurately, so we assume the correlation between attributes in the prior is accurate and we only need to learn some roots in the graphical model. In practice, we select some one-way marginal queries for the workload $\Wgroup$ such that those attributes are approximately independent according to the public dataset. 

\eat{
Ideally, if we have more marginal queries, we have a better estimation of the distribution, but the sensitivity of the entire workload is high. Besides, if the domain of a marginal query is high, we will also end up injecting too much noise when the size of data is not sufficiently large. Therefore, we consider re-grouping the data by STATE, which allows us to ask high-way marginal queries with low error. In contrast, we only ask one-way marginal queries for each PUMA-YEAR group. \am{Could you have used HDMM for the select step? Why or Why not?}\yt{HDMM is data independent}

The selection of high-way marginal or one-way marginals is based on the pair-wise mutual information \cite{steuer2002mutual}. If there exists some attributes that all have high mutual information with each other, we select them as a high-way marginal query. For select one-way marginal queries, we use an opposite strategy. For attributes that have low mutual information with each other, we select each as a one-way marginal query.  \am{I think this can be described more succinctly.}
}

\eat{
\paratitle{Random workload selection.} In Algorithm \ref{alg:iterative} "Iterative Distribution Estimation", we randomly select a partial workload $W_s$ from $W$ such that $\norm{W_s}_1 =1$. In practice, we treat every marginal query in the workload as a candidate partial workload and select them in a fixed sequence. 
}

\paratitle{Iteration parameter $T$.} The iteration parameter $T$ should be set as a large number to ensure the iterative distribution estimation converges. In practice, we set $T$ as two times the number of marginal queries in the workload during the contest due to the run-time limit. However, if the code implementation is faster, $T$ should be set larger.


\subsection{Result}

In the contest, the final evaluation is based on a secret dataset that is not publicly accessible. We evaluate our algorithm locally on a private dataset that is the same as the public dataset. Notice that the algorithm doesn't know that they are the same. Figure \ref{fig:result} shows the average $L_1$ error of random 50 2-way marginals for each PUMA-YEAR group under $\epsilon=10$, which is around 0.2 to 0.4. There is also no bias penalty for the synthetic dataset. The result is based on one run of the algorithm.

\begin{figure}[H]
    \centering
    \begin{subfigure}[b]{0.49\textwidth}
        \includegraphics[width=\linewidth]{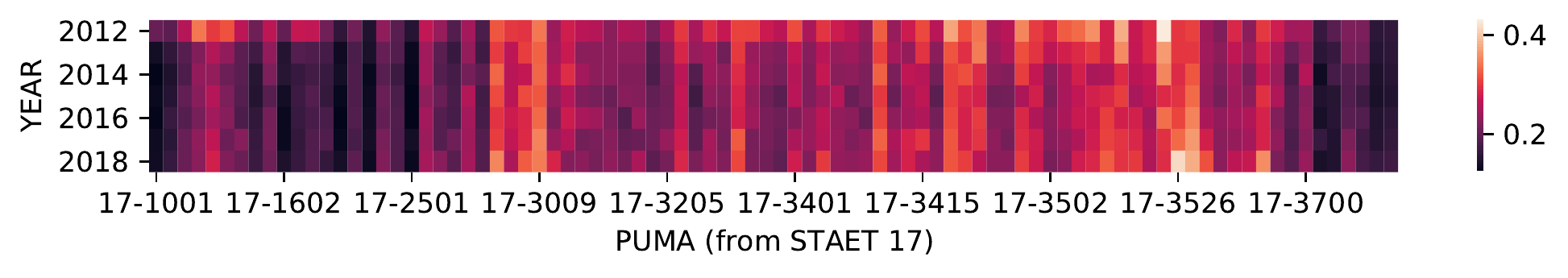}
    \end{subfigure}
    \begin{subfigure}[b]{0.49\textwidth}
        \includegraphics[width=\linewidth]{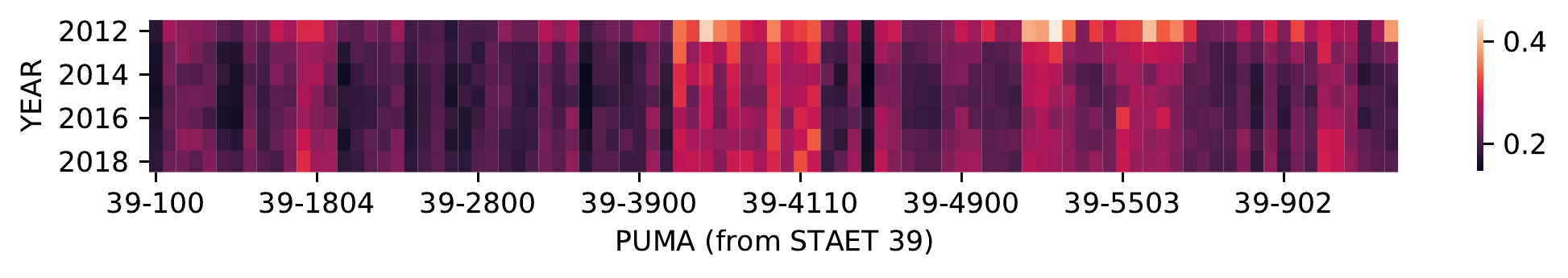}
    \end{subfigure}

    \caption{Average L1 error of random 50 2-way marginals for each PUMA-YEAR group under $\epsilon = 10$. }
    \label{fig:result}
\end{figure}

\clearpage

\nobalance
\bibliographystyle{unsrt}
\bibliography{refs}

\clearpage
\appendix

\section{Algorithms}

\begin{algorithm}[H] 
\caption{Synthetic Data Generation} 
\begin{algorithmic}[1]
\Require A workload $\Wstate$ for the state-level data, a workload $\Wgroup$ for the group-level data, the private data $D$ as a set of $\{(puma, year, \f)\}$ where $\f$ is the frequency distribution for the group $(puma, year)$, the public prior $\pp$, the individual stability $\Delta$, the privacy budget $\epsilon$, and an iteration parameter $T$. 
\Ensure A synthetic dataset.
\Function{Prior\_Update}{$W, \f, \pp, T, \epsilon_0$}
    \State $\tilde{n} \gets \norm{\f}_1 + Laplace(\Delta / \epsilon_0)$
    \State $\tilde{n} \gets \round{\max\{\tilde{n}, 0 \}}$
    \State $\nY \gets W \f + Laplace(\norm{W}_1  \Delta / \epsilon_0)$
    \State $\nY \gets \nY / \tilde{n}$
    \State $\pp' \gets$ IDE($W, \nY, \pp, T$) 
    \Comment{See Algorithm \ref{alg:iterative}}
    \State \Return $\pp'$, $\tilde{n}$
\EndFunction
\State $D' \gets \emptyset$
\State $\epsilon_0 \gets \epsilon / 4$
\For{$S \gets$ GROUP-BY($D$, STATE)}
    \State $\f_s \gets \sum_{(puma, year, \f) \in S} \f$
    \State $\pp_s, \tilde{n}_s \gets$ \textproc{Prior\_Update}($\Wstate, \f_s, \pp, T, \epsilon_0$)
    \For{$(puma, year, \f_g) \in S$}
    \State $\pp_g, \tilde{n}_g \gets$ \textproc{Prior\_Update}($\Wgroup, \f_g, \pp_s, T, \epsilon_0$)
        \State $\f' \gets$ sample $\tilde{n}_g$ tuples from $\pp_g$
        \State $D' \gets D' \cup (puma, year, \f')$
    \EndFor
\EndFor
\State \Return $D'$.
\end{algorithmic}
\label{alg:sdg}
\end{algorithm}

\section{Theorem Proof}

\lreu* 
\begin{proof}(sketch)
Consider the value of $\qq(v_i)$ in the optimal solution. For a partial workload $W_s$ with $L1$-norm equal to 1, the optimal solution according to the Lagrange function has a simple form:
\[
\qq(v_i) = \pp(v_i) \exp\left[ -1 -\mu - \lambda_j \right]
\]
for some $j$ such that $W_s[j, i] = 1$. We discuss the case latter when there is no $j$ such that this condition holds. Now we fix $j$, from the constraints we have 
\[
\sum_{i \mid W_s[j, i] = 1} q(v_i) = y_s[j] 
\]
. Together we have
\[
\sum_{i \mid W_s[j, i] = 1} \pp(v_i) \exp\left[ -1 -\mu - \lambda_j \right]= y_s[j]
\]
, which implies
\[ \exp\left[ -1 -\mu - \lambda_j \right]= \frac{y_s[j]}{\sum_{i \mid W_s[j, i] = 1} \pp(v_i)}
\]. Take this back to the first equation, we have
\[
\qq(v_i) = \pp(v_i) \frac{y_s[j]}{\sum_{i \mid W_s[j, i] = 1} \pp(v_i)}
\]
When $\sum_{i \mid W_s[j, i] = 1} \pp(v_i) = 0$, we divide $y_s[j]$ equally for each $q(v_i)$. 

For $\qq(v_i)$ in the case that there is no $j$ such that $W_s[j, i] = 1$, the optimal solution has this form:
\[
\qq(v_i) = \pp(v_i) \exp\left[ -1 -\mu \right]
\]
Similar analysis can be applied to this case. Notice that in this case, from the constraints we know that the sum of all such $q(v_i)$ should be equal to $1 - \norm{y_s}_1$.
\end{proof}

\sdg*
\begin{proof} (sketch)
Each call of \textproc{Prior\_Update} runs two Laplace mechanisms sequentially, each satisfy $\epsilon_0$-DP. For each PUMA-YEAR group, it is taken (partially) as the input for \textproc{Prior\_Update} twice. According to sequential composition rule and parallel composition rule, the entire algorithm satisfies $4 \epsilon_0$-DP, which is also $\epsilon$-DP.
\end{proof}
\end{document}